\documentclass[12pt]{article}
\usepackage[english]{babel}
\usepackage[utf8x]{inputenc}
\usepackage{amsmath}
\usepackage{graphicx}
\usepackage[x11names]{xcolor}
\usepackage[colorinlistoftodos]{todonotes}
\usepackage{comment}

\usepackage{amsmath,amsthm}  
\usepackage{pst-node}        
\usepackage[crop=off]{auto-pst-pdf} 
\def\Property{\ensuremath{\hat{P}}}  
\usepackage[x11names]{xcolor}
\usepackage{todonotes}

\newcommand{\G}{\mbox{\bf G\,}}
\newcommand{\X}{\mbox{\bf X\,}}

\newcommand{\W}{\mbox{\bf \,W\,}}

\newtheorem{thm}{Theorem}

\newtheorem{cor}{Corollary}

\begin{document}

\begin{titlepage}

\newcommand{\HRule}{\rule{\linewidth}{0.5mm}} 

\center 


\textsc{\LARGE Universit\'{e} du Qu\'{e}bec \`{a} Chicoutimi}\\[1.5cm] 
\textsc{\Large Department of Computer Science and Mathematics }\\[0.5cm] 
\textsc{\large Laboratoire D'informatique Formelle}\\[0.5cm] 


\HRule \\[0.4cm]
{ \huge \bfseries Runtime Enforcement With Partial Control}\\[0.4cm] 
\HRule \\[1.5cm]


\begin{minipage}{0.4\textwidth}
\begin{flushleft} \large
Rapha\"{e}l \textsc{Khoury} 
\end{flushleft}
\end{minipage}
~
\begin{minipage}{0.4\textwidth}
\begin{flushright} \large
Sylvain \textsc{Hall\'{e}} 
\end{flushright}
\end{minipage}\\[2cm]



{\large \today}\\[2cm] 




\vfill 

\end{titlepage}

\begin{abstract}
This study carries forward the line of enquiry that seeks to characterize precisely which security policies are enforceable by runtime monitors. In this regard, Basin et al.\ recently refined the structure that helps distinguish between those actions that the monitor can potentially suppress or insert in the execution, from those that the monitor can only observe. In this paper, we generalize this model by organizing the universe of possible actions in a lattice that naturally corresponds to the levels of monitor control. We then delineate the set of properties that are enforceable under this paradigm and relate our results to previous work in the field. Finally, we explore the set of security policies that are enforceable if the monitor is given greater latitude to alter the execution of its target, which allows us to reflect on the capabilities of different types of monitors.
\end{abstract}

\section{Introduction}\label{sec:introduction}

Runtime monitoring is an approach to enforcing security policies that seeks to allow untrusted code to run safely by observing its execution and reacting as needed to prevent a violation of a user-supplied security policy.  This method of ensuring the security of code is rapidly gaining acceptance in practice and several implementations exist \cite{COSREVIEW12}.  One question seems to recur frequently in multiple studies: exactly which set of properties are \textit{monitorable}, in the sense that they are enforceable by monitors. Previous research has identified several factors that can affect the set of security policies enforceable by monitors. These include the means at the disposal of monitors to react to a potential violation of the security policy \cite{MoreEnforce}, the availability of statically gathered data about the target program's possible executions \cite{ChabotJournal,MoreEnforce}, memory and computability constraints \cite{fong,Computability2} etc.

One specific aspect that can have a considerable impact on the monitor's expressiveness is its ability to either suppress certain actions performed by the target program from occurring during the execution (while allowing the remainder of the execution to continue unaffected) or to insert additional events in an ongoing execution. These abilities, when available, extend the monitor's enforcement power considerably. Indeed, a lower bound on the enforcement power of monitors is given by Schneider \cite{enforceable} who shows that the set of properties enforceable by a monitor whose only possible reaction to a potential violation of the desired security policy is to abort the execution coincides with the set of safety properties. Conversely, Ligatti et al.\ \cite{nonsafetyJournal} consider the case of a monitor with an unlimited ability to delay  any event performed by the target program until it has ascertained that its occurrence in the execution would not violate the security policy.  In effect, the monitor is simulating the execution of the program until it is certain that the behaviour it has so far witnessed is correct.  When behaving in this manner, the monitor can enforce a vast range of security properties, termed the set of infinite renewal properties, which   includes all safety policies, some liveness policies and some policies that are neither safety nor liveness.

Yet, it may not be realistic to assume that the monitor has an unlimited ability to simulate the execution of the target program.
Indeed,  as Ligatti et al.\  \cite{nonsafetyJournal} point out: ``\textit{[O]ur model assumes that security automata have the same computational capabilities as the system that observes the monitor's output. If an action violates this assumption by requiring an outside system in order to be executed, it cannot be feigned (i.e., suppressed) by the monitor. For example, it would be impossible for a monitor to feign sending email, wait for the target to receive a response to the email, test whether the target does something invalid with the response, and then decide to undo sending email in the first place. Here, the action for sending email has to be made observable to systems outside of the monitor's control in order to be executed, so this is an unsuppressible action. [...] Similarly, a system may contain actions uninsertable by monitors because, for example, the monitors [...] lack access to secret keys that must be passed as parameters to the actions. In general, environmental factors beyond the control of the monitor may give rise to actions that are unsuppressible or uninsertable.}''
The set of infinite renewal should thus be seen as an upper bound to the enforcement power of monitors.

To this end, Basin et al.\ propose a middle ground \cite{enforceableRevisited}. They partition the set of possible program actions in two disjoint subsets: a set of controllable actions, which the monitor may freely suppress from the execution, and a set of observable actions, whose occurrence the monitor can only observe. This allows for a more precise characterization to the set of monitorable properties. Section \ref{sec:preliminaries} will discuss these concepts in more detail.

In Section \ref{sec:partial}, we further generalize this analysis by organizing the set of possible actions along a lattice that distinguishes between four types of atomic actions, namely \emph{controllable} actions (which a monitor can insert or block from an execution), \emph{insertable} actions (which a monitor can add to the execution but not suppress), \emph{suppressible} actions (the converse) and \emph{observable} actions (which the monitor can only observe). We then delineate the set of properties that are enforceable under this paradigm and relate our results to previous work in the field.

Finally, we explore in Section \ref{sec:equivalence} the set of security policies that are enforceable if the monitor is given greater latitude to alter the execution of its target, rather than be bounded to return a syntactically identical execution sequence if the original execution is valid. In particular, we consider a monitor which can \textit{add} any action into the execution, but cannot prevent any action from occurring if the target program requests it. We also consider a monitor can \textit{remove} potentially malicious actions  performed by the target program but cannot add any action to the execution. We show how both can be handled by our model by simply considering a different equivalence relation between traces.

\section{Preliminaries}\label{sec:preliminaries}

\subsection{Executions}
Executions are modelled as sequences of atomic actions taken from a finite or countably infinite set of actions $\Sigma$.  The empty sequence is noted $\epsilon$, the set of all finite length sequences is noted $\Sigma^*$, that of all infinite length sequences is noted $\Sigma^\omega$, and the set of all possible sequences is noted $\Sigma^\infty = \Sigma^\omega \cup \Sigma^*$. 
Let $\tau \in \Sigma^*$ and $\sigma \in \Sigma^\infty$ be two sequences of actions.
We write $\tau;\sigma$ for the concatenation of $\tau$ and $\sigma$. We say that $\tau$ is a prefix of $\sigma$ noted $\tau \preceq \sigma$, or equivalently $\sigma\succeq\tau$  $\mathit{iff}$ there exists a sequence $\sigma'$ such that $\tau;\sigma'=\sigma$. Let $\tau, \sigma \in \Sigma^\infty$ be a sequence, we write acts $(\sigma)$ for the set of actions  present in $\sigma$. We write $res_{\Property}(\sigma)$ for the residual of $\Property$ with regard to $\sigma$, i.e. the set of sequence $S\subseteq\Sigma^\infty$ s.t. $\forall \tau\in S :\sigma;\tau\in\Property$. 
Finally, let $\tau, \sigma \in \Sigma^\infty$, $\tau$ is  said to be a suffix of $\sigma$ iff there exists a $\sigma'\in\Sigma^*$ such that $\sigma=\sigma';\tau$.

Following \cite{nonsafetyJournal}, if $\sigma$ has already been quantified, we freely write $\forall \tau\preceq\sigma$ (resp. $\exists \tau\preceq\sigma$) as an abbreviation to $\forall \tau \in \Sigma^*:\tau\preceq\sigma$ (resp. $\exists \tau \in \Sigma^*:\tau\preceq\sigma$).  Likewise, if $\tau$ has already been quantified,  $\forall \sigma \in \Sigma^\infty:\sigma\succeq\tau$ (resp. $\exists \sigma \in \Sigma^\infty:\sigma\succeq\tau$) can be abbreviated as $\forall \sigma\succeq\tau$ (resp. $\exists \sigma\succeq\tau$).

Let $\tau$ be a sequence and $a$ be an action, we write $\tau\backslash a$ for the left cancellation of $a$ from $\tau$, which is defined as the removal from $\tau$ of the first occurrence of $a$.  Formally:

\begin{displaymath}
a;\tau\backslash a'= \begin{cases}
                          \tau & \hbox{if } a=a'; \\
                          a;(\tau\backslash a') & \hbox{otherwise}
                        \end{cases}
\end{displaymath}

Observe that $\epsilon\backslash a=\epsilon$.  Abusing the notation, we write $\tau\backslash \tau'$ to denote the sequence obtained by left cancellation of each action of $\tau'$ from $\tau$. Formally, $\tau\backslash a ;\tau'=(\tau\backslash a )\backslash \tau'$. For example, $abcada \backslash daa = bca$.

A finite word $\tau\in\Sigma^*$ is said to be a subword of a word $\omega$, noted $\tau\triangleleft_{\Sigma}\sigma$, iff $\tau = a_0a_1a_2a_3...a_k$ and $\omega=\sigma_0 a_0 \sigma_1 a_1 \sigma_2 a_2 \sigma_3 a_3... \sigma_k a_k \upsilon $ with $\sigma_0, \sigma_1,\sigma_2... \in \Sigma^*$ and $\upsilon\in \Sigma^\infty$.  Let $\tau,\sigma$ be sequences form $\Sigma^*$. We write $cs_\tau(\sigma)$ to denote the longest subword of $\tau$ which is also a subword of $\sigma$. For any $\tau \neq \epsilon$, $\tau.last$ denotes the last action of sequence $\tau$.


\subsection{Security Policies and Security Properties}

A security policy $P$ is a property iff it can be characterized as a set of sequences for which there exists a decidable predicate $\Property$ over the executions of $\Sigma^\infty:\Property(\sigma)$ iff $\sigma$ is in the policy \cite{enforceable}. In other words, a property is a policy for which the membership of any sequence can be determined by examining only the sequence itself \footnote{Security policies whose enforcement necessitates the examination of multiples execution sequences, such as noninterference policies, are not generally enforceable by monitors.}. Such a sequence is said to be \textit{valid} or to \textit{respect} the property. Since, by definition, all policies enforceable by monitors are properties, P and \Property\ are used interchangeably in our context. Additionally, since the properties of interest represent subsets of $\Sigma^\infty$, we follow the common usage in the literature and freely use \Property\ to refer to these sets.

A number of classes of properties have been defined in the literature and are of special interest in the study of monitoring. First are {\it safety} properties \cite{safety}, which proscribe the occurrence of a certain ``bad thing'' during the execution. Formally, let $\Sigma$ be a set of actions and  \Property\ be a property. \Property\ is a {\it safety} property iff
\begin{equation}
\tag{safety}
 \forall\sigma\in
 \Sigma^\infty:\neg\Property(\sigma)\Rightarrow\exists\sigma'\preceq\sigma:
\forall\tau\succeq\sigma':\neg\Property(\tau)
 \end{equation}

Informally, this states that any sequence does not respect the security property if there exists a prefix of that sequence from which any possible extension does not respect the security policy.   This implies that a violation of a safety property is irremediable: once a violation occurs, nothing can be done to correct the situation.

Alternatively, a {\it liveness} property \cite{defliv} is a property prescribing that a certain ``good thing'' must occur in any valid execution. Formally, for an action set $\Sigma$ and a property \Property, \Property\ is a liveness property iff
\begin{equation}
\tag{liveness}
  \forall\sigma\in \Sigma^*:\exists\tau \in \Sigma^\infty: \tau
  \succeq\sigma \wedge \Property(\tau)
 \end{equation}

Informally, the definition states that a property is a liveness property if any finite sequence can be extended into a valid sequence.

Another class of security properties that are of interest is that of {\it renewal} properties\cite{nonsafetyJournal}. A property is in renewal if every infinite valid sequence has infinitely many valid prefixes, while every infinite invalid sequence has only finitely many such prefixes. Observe that every property over finite sequences is in infinite renewal.  The set of renewal properties is equivalent to the set of response properties in the safety-progress classification \cite{safPro}.
\begin{equation}
\tag{renewal}
 \forall\Property\subseteq\Sigma^\omega:\Property(\sigma)\Leftrightarrow\exists \sigma'\preceq\sigma:\exists\tau\preceq\omega:\tau\succeq\sigma'\wedge\Property(\tau).
\end{equation}

It is often useful to restrict our analysis to properties for which the empty sequence $\epsilon$ is valid. Such properties are said to be \textit{reasonable} \cite{nonsafetyJournal}.  Formally,
\begin{equation}
\tag{reasonable}
   \forall\Property\subseteq\Sigma^\infty: \Property(\epsilon)\Leftrightarrow\Property{} \textit{ is reasonable}
 \end{equation}

In the remainder of this paper, we will only consider reasonable properties.  Furthermore, in order to avoid having the main topic of this paper be sidestepped by decidability issues, will consider that $\Property(\sigma)$ is decidable for all properties and all execution sequences. Likewise, we also consider that other predicates or functions over sequences are decidable.

\subsection{Security Property Enforcement}

Finally, we need to provide a definition of what it means to ``enforce'' a security property \Property.  A number of possible definitions have been suggested. The most widely used is effective$_{\cong}$ enforcement \cite{MoreEnforce}. Under this definition, a property is effectively$_{\cong}$ enforced iff the following two criterion are respected.
\begin{enumerate}
  \item Soundness: All observable behaviours of the target program respect the desired property, i.e.\ every output sequence is present in the set of executions defined by \Property.
  \item Transparency: The semantics of valid executions is preserved, i.e.\ if the execution of the unmonitored program already respects the security property, the monitor must output  an equivalent sequence, with respect to an equivalence relation $\cong\subseteq\Sigma^\infty\times\Sigma^\infty$.
\end{enumerate}

Syntactic equality is the most straightforward equivalence relation, and the one that has been the most studied in the literature.  It models the behaviour of a monitor that enforces the desired property by suppressing (or simulating) part of the execution, only allowing it to be output  when it has ascertained that the execution up to that point is valid. In section \ref{sec:equivalence}, we consider two alternative notions of equivalences, which can be used to characterize alternative behaviours on the part of the enforcement mechanism.
\subsection{Related Work}\label{sec:related}

Initial work on the question of delineating which security policies are or are not enforceable by monitor was performed by Schneider \cite{enforceable}. He considered the capabilities of a monitor that observes the execution of its target, with no knowledge of its possible future behaviour and no means to affect the target except by aborting the execution. Each time the target program attempts to perform an action, the monitor has to either accept it immediately, or abort the execution. Under these constraints, the set of properties enforceable by monitors coincides with the set of \textit{safety} properties.

Ligatti et al.\ \cite{editauto} extend Schneider's modelling of monitors along three axes:
\begin{enumerate}
\item According to the means at the disposal of the monitor to react to a potential violation of the security policy. These include truncating the execution, inserting new actions into the execution, suppressing some part of the executions or both inserting and suppressing actions
\item According to the availability of statically gathered data describing the target program possible execution paths
\item According to how much latitude the monitor is given to alter executions that already respect the security policy.
\end{enumerate}
By combining these three criteria, they build a rich taxonomy of enforceable properties, and contrast the enforcement power of different types of monitors.

Basin et al.\ \cite{enforceableRevisited} generalize Schneider's model by distinguishing between observable actions, whose occurrence the monitor cannot prevent, and controllable actions, which the monitor can prevent from occurring by aborting the execution.

The enforcement power of monitor operating with memory constraints is studied in \cite{fong}, \cite{TTD06-a,TTD06-b,TTD08} and \cite{beauquier}.

The computability constraints that can further restrict a monitor's enforcement power are discussed in \cite{Computability2,CompuARM}; that of monitors relying upon an \textit{a priori} model of the program's possible behaviour is discussed in \cite{ChabotJournal} and \cite{editauto}.

Falcone et al.\ \cite{safProEnf1,safProEnf3} show that the set of infinite renewal properties coincides  with the union of four of the 6 classes of the safety-progress classification of security properties \cite{safPro}. Khoury and Tawbi \cite{KhouryEquiv,KT2012} and Bielova et al.\ \cite{BielovaNordsec,WhatMeanArt,DBLP:conf/essos/BielovaM11} further refine the notion of enforcement by suggesting alternative definitions of enforcement. In \cite{runtimeResults} Ligatti and Reddy introduced an alternative model, the mandatory-result automaton.  This model distinguishes between the action set of the target and that of the system with which it interacts. This distinction makes it easier to study the interaction between the target program, the monitor and the system. A thorough survey of the question of enforceable properties by monitors is provided in \cite{COSREVIEW12}.

\section{Monitoring With Partial Control}\label{sec:partial}

The previous works each consider monitors where actions belong to particular sets. For example, Schneider's model assumes that all actions can be suppressed by the monitor; conversely, Ligatti et al.\ assume that all actions can be indefinitely delayed. Basin et al.\ propose a middle ground where every action can either be freely suppressed, or can only be observed. In this section, we define a generalized model of actions where each of these works becomes a particular case. We then study what properties are enforceable in this generalized model.

\subsection{A Lattice of Actions}

We organize the set of possible actions along a lattice that distinguishes between four types of atomic actions :  namely controllable actions ($\mathcal{C}$), insertable actions ($\mathcal{I}$), suppressible actions ($\mathcal{D}$ ---for delete) and observable actions ($\mathcal{O}$), as is shown in Figure \ref{fig:lattice}.
\begin{itemize}
  \item Controllable actions ($\mathcal{C}$) are the basic actions such as opening a file or sending data on the network, which the monitor can either insert into the execution or prevent from occurring if they violate the security policy. In Ligatti's model, all actions are controllable.
  \item Insertable actions ($\mathcal{I}$) can be added into the execution but not suppressed if they are present. An example of an insertable action is an additional delay before processing a request to bring it unto compliance with a resource usage policy. The monitor may add such actions to the execution sequence but cannot remove them if the target program executes them.
  \item Suppressible actions ($\mathcal{D}$ ---for delete) are those actions that the monitor can prevent from occurring, but cannot insert in the execution if the target program  does not request them. Sending an email, decrypting a file or receiving a user input are all examples of suppressible actions. In Schneider's model, all actions are suppressible.
  \item Observable actions $\mathcal{O}$ can only be observed by the monitor, which can neither insert them in the execution when they are not present  nor suppress them if they occur. In Basin et al.'s model, all actions are either suppressible or observable.
\end{itemize}

As the examples above illustrate, we believe that the lattice model we propose is a more realistic description of the reality encountered by the monitor, and will thus allow a more precise characterization of the set of enforceable properties.  Observe than the monitor may only abort the execution if the next action is in $\mathcal{C}\cup\mathcal{D}$. Moreover, the sets $\mathcal{C},\mathcal{I},\mathcal{D},\mathcal{O}$ are disjoint and that the universe of possible program actions is $\Sigma =  \mathcal{C}\cup\mathcal{I}\cup\mathcal{D}\cup\mathcal{O}$.

\begin{figure}
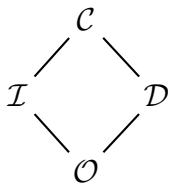

\centering
\psset{nodesep=3pt}
\newpsstyle{DblDash}{linestyle=dashed, dash=1pt 1.5pt, doubleline}
\begin{psmatrix}[mnode=r,colsep=0.6,rowsep=0.5]
&[name=12] $\mathcal{C}$&\\
[name=21] $\mathcal{I}$&&[name=23] $\mathcal{D}$\\
&[name=32] $\mathcal{O}$&
\ncline{12}{21}
\ncline{12}{23}
\ncline{21}{32}
\ncline{23}{32}
\end{psmatrix}
\caption{The lattice of possible actions}
\label{fig:lattice}
\end{figure}


The following notation is useful for comparing different enforcement mechanisms. Let $\Sigma$  be a universe of actions and let $\mathcal{L}$ be a lattice over the set $\Sigma$ as described in section 2. Let $\mathcal{S}\subseteq\Sigma^{\infty}$  stand for a subset of possible execution sequences and let $\cong\subseteq \Sigma^\infty \times \Sigma^\infty$ be an equivalence relation. We write   $\mathcal{L}^{\mathcal{S}}$-enforceable$_{\cong}$ to denote the set of proprieties that are enforceable$_{\cong}$ by a monitor when the set of possible sequences is $\mathcal{S}$ and the set of possible actions is organized alongside lattice $\mathcal{L}$.

Let $\mathcal{L}$ be a lattice as described above. We write $\mathcal{L_O}$ (resp. $\mathcal{L_I}$, $\mathcal{L_D}$, $\mathcal{L_C}$) for the set  $\mathcal{O}$ (resp. $\mathcal{I}$, $\mathcal{D}$, $\mathcal{C}$) in $\mathcal{L}$.
We write $\mathcal{L}=\langle\Sigma_1,\Sigma_2,\Sigma_3,\Sigma_4\rangle$ for the lattice where $\mathcal{L_O}=\Sigma_1$, $\mathcal{L_I}=\Sigma_2$, $\mathcal{L_D}=\Sigma_3$ and  $\mathcal{L_C}=\Sigma_4$.
Let $A, B \in \{ \mathcal{O,I,D,C}\}$ and let $a \in \Sigma$,
we write $\mathcal{L}_{A\xrightarrow{a} B}$  to indicate the lattice $\mathcal{L'}$ defined such that $\mathcal{L'}_A = \mathcal{L}_A\backslash \{a\}, \mathcal{L'}_B  = \mathcal{L}_B\cup \{a\}$ and $\forall C \in \{ \mathcal{O,I,D,C}\}: C\notin \{A, B\} \Rightarrow \mathcal{L'}_C = \mathcal{L}_C$. In other words, $\mathcal{L}_{A\xrightarrow{a}B}$ is the lattice built by moving only element $a$  from set $A$ to another $B$, leaving all other sets unchanged.


\subsection{Enforceable Properties}

We begin by reflecting on the set of properties that are $\mathcal{L}^{\Sigma^\infty}$-enforceable$_{=}$, i.e.\ properties that are enforceable if the monitor is bounded to output any valid sequence exactly as it occurs (with syntactic equality as the equivalence relation between valid inputs and the monitor's output).  This is the enforcement paradigm that has been the most studied in the literature.  A monitor that seeks to enforce a property in this manner may take any one of three strategies, depending on the desired property and the ongoing execution, and the set of  $\mathcal{L}^{\Sigma^\infty}$-enforceable$_{=}$ properties can be derived by combining the three.

First, using a model in which every action is controllable, Ligatti et al.\ argued that a monitor can enforce any reasonable renewal property by suppressing the execution until a valid prefix is reached at which point the monitor can output the suffix of the execution it has previously suppressed.  We generalized the definition of renewal as follows:
\begin{multline}
\tag{$\mathcal{L}$-Renewal}
  \forall\sigma\in \Sigma^\infty: \Property(\sigma)\Leftrightarrow \forall \sigma'\preceq \sigma: (\exists \tau\preceq\sigma:\sigma'\preceq\tau\wedge\Property(\tau) \wedge \\ \forall\tau'\preceq\tau:\tau'\succeq\sigma'\Rightarrow\tau'.last\in{\mathcal{C}}).
\end{multline}

Observe that the definition now applies to all sequences in $\Sigma^\infty$, rather than just to infinite sequences. The second half of the equation always evaluates to true if all the actions are controllable (Ligatti's model) and always evaluates to false if all the actions are suppressible or observable (Schneider or Basin's models).

Second, Ligatti et al.\ observe that a property is enforceable if there exists a prefix beyond which there is only one valid extension. In that case, the monitor can abort the execution and output that sequence. This allows some nonsafety properties to be monitored. Ligatti et al.\ refer to this case as the ``corner case'' of effective enforcement.  We generalize this case as follows:
\begin{multline}
\tag{$\mathcal{L}$-Corner case}
\forall\sigma\in \Sigma^\infty: \Property(\sigma)\vee \exists \sigma'\preceq\sigma :\forall \sigma';\tau\succeq\sigma':\Property(\sigma';\tau)\Rightarrow \sigma';\tau=\sigma \wedge \\
\sigma'.last\in \mathcal{D\cup C} \wedge acts(\tau)\subseteq  \mathcal{I\cup C}.
\end{multline}

Once again, observe that the equation restricts all sequences, rather than only infinite ones. This equation always evaluates to false in the Schneider-Basin model and the second conjunct always evaluates to true in Ligatti's model.

Finally, the monitor can simply abort the execution if it is irremediably invalid. This is a generalization of the set of safety properties to our framework.
\begin{multline}
\tag{$\mathcal{L}$-Safety} 
  \forall\sigma\in \Sigma^\infty:\neg \Property(\sigma)\Rightarrow \exists \tau;a\preceq\sigma :\Property(\tau) \wedge a\in\mathcal{D\cup C} \\
  \wedge \neg\exists \tau'\succeq\tau:\Property(\tau').
\end{multline}

The set of $\mathcal{L}^{\Sigma^\infty}$-enforceable$_{=}$ properties  can now be stated. The definition is not simply the conjunction of the tree preceding set as it must take into account the possibility that enforcement might begin by suppressing and reinserting part of the execution, and then abort the execution using of one of the other methods.
\begin{thm}\label{thm:Master}
\begin{multline*}
\Property{} \in \mathcal{L}^{\Sigma^\infty}-\mbox{enforceable}_{=}  \Leftrightarrow \forall \sigma\in \Sigma^\infty:\\
(\Property(\sigma)\Leftrightarrow \forall\sigma'\preceq\sigma:\exists \tau\preceq\sigma:\sigma'\preceq\tau\wedge\Property(\tau)\wedge (\forall\tau'\preceq\tau:\neg\Property(\tau')\Rightarrow \tau'.last\in \mathcal{C})\vee \\
(\exists \tau;a\preceq\sigma: a\in \mathcal{D} \cup \mathcal{C}  \wedge \forall\tau'\preceq\tau;a:\neg\Property(\tau')\Rightarrow \tau'.last\in \mathcal{C} \wedge
(\forall \tau'\succeq\tau;a:\Property(\tau')\Rightarrow \\
\tau'=\sigma  \vee  (\Property(\tau)\wedge \neg\exists \tau'\succeq \tau:\Property(\tau') \wedge  acts(\sigma \backslash \tau' ;a) \in C \cup I  ) )) )
\end{multline*}
\end{thm}
\begin{proof}
See Appendix A for all proofs to theorems and corollaries.
\end{proof}

This definition narrows the set of monitorable properties somewhat, compared with previous work. For example, bounded availability is given in \cite{nonsafetyJournal} as an example of a monitorable policy. In fact, it is monitorable only if every action  that occurs between the acquisition and release of a protected resource is controllable. Likewise, the ``\textit{no send after read}'' property described in \cite{enforceable} is only enforceable if every action which might violate the property is deletable or controllable.

While the set-theoretic characterization of enforceable property is somewhat involved, an LTL property can characterize such properties.

\begin{thm}
Let \textit{valid} be a predicate identifying a valid sequence:
\[
\mbox{valid}(\sigma)\Leftrightarrow\Property(\sigma)
\]
Let \textit{cc} be a predicate that identifies a sequence in the $\mathcal{L}$-corner case:
\begin{multline*}
\mbox{cc}(\sigma)\Leftrightarrow\Property(\sigma)\wedge \exists \sigma'\preceq\sigma :\forall \sigma';\tau\succeq\sigma':\Property(\sigma';\tau)\Rightarrow \\ \sigma';\tau=\sigma \wedge\sigma'.last\in \mathcal{D\cup C} \wedge acts(\tau)\subseteq  \mathcal{I\cup C}
\end{multline*}
Let $C$ be a predicate identifying a sequence ending on a controllable action, and \textit{D} a predicate that identifies a sequence ending on a deletable action:
\begin{eqnarray*}
\mbox{C}(\sigma)& \Leftrightarrow & \sigma.last\in \mathcal{C} \\
\mbox{D}(\sigma)& \Leftrightarrow & \sigma.last\in \mathcal{D}
\end{eqnarray*}
Then we have:
\begin{multline*}
\Property{} \in \mathcal{L}^{\Sigma^\infty}$-enforceable$_{=} \Leftrightarrow  \forall \sigma\in \Sigma^\infty:  \\
\G ( C \W \mbox{valid}) \vee (C \W \mbox{valid} \vee \X ((D \vee C) \wedge (\G \neg \mbox{valid} \vee \mbox{cc})))
\end{multline*}
\end{thm}

In this theorem, we assume that any execution sequence from $\Sigma^\infty$ is possible. In other words, at each step of the execution, the monitor must assume that any action from $\Sigma$ is a possible next action, or that the execution of the target program could stop. This is called the uniform enforcement context. The monitor often operates in a context where it knows that certain executions are impossible (the nonuniform context).
This situation occurs when the monitor benefits from a static analysis of its target, that provides it with a model of the target's possible behaviour.  We can adapt the above theorem to take into account the fact that the monitor could operate in a nonuniform context.

\begin{thm}
\begin{multline*}
\Property{} \in \mathcal{L}^{\mathcal{S}}$-enforceable$_{=}$ $ \Leftrightarrow \forall \sigma\in \mathcal{S}:\\
(\Property(\sigma)\Leftrightarrow \forall\sigma'\preceq\sigma:\exists \tau\preceq\sigma:\sigma'\preceq\tau\wedge\Property(\tau)\wedge (\forall\tau'\preceq\tau:\neg\Property(\tau')\Rightarrow \\
(\tau'.last\in \mathcal{C} \wedge \tau \in \mathcal{S}) )\vee \\ (\exists \tau;a\preceq\sigma: a\in  \mathcal{D} \cup \mathcal{C} \wedge \forall\tau'\preceq\tau;a:(\neg\Property(\tau')\wedge \tau' \in \mathcal{S})\Rightarrow \tau'.last\in \mathcal{C} \wedge \\ (\forall \tau'\succeq\tau;a:\Property(\tau')\Rightarrow \tau'=\sigma  \vee  (\Property(\tau)\wedge \neg\exists \tau'\succeq \tau:\\
\Property(\tau')\wedge \tau'\in \mathcal{S} \wedge  acts(\sigma \backslash \tau' ;a) \in \mathcal{C} \cup \mathcal{I}  ) )) )
\end{multline*}
\end{thm}

We can now restate the results of previous research in our new formalism.
\begin{thm}
(from \cite{enforceable}) If $\mathcal{L} =\langle\emptyset,\emptyset,\Sigma,\emptyset\rangle$ then $\mathcal{L}^{\Sigma^\infty}$enforceable$_{=}$ is  Safety.
\end{thm}

\begin{thm}
(from \cite{nonsafetyJournal}) If $\mathcal{L} =\langle\emptyset,\emptyset,\emptyset,\Sigma\rangle$ then $\mathcal{L}^{\Sigma^\infty}$enforceable$_{=}$ is the union of Infinite Renewal  and the corner case.
\end{thm}


\subsection{Additional Results}
\begin{cor}\label{corr:1}
Let $\mathcal{L}$ be a lattice over a set of actions $\Sigma$ as described above and let $\Sigma_1,\Sigma_2\in\{\mathcal{O},\mathcal{D},\mathcal{I},\mathcal{C}\}$ and $\Sigma_1 \sqsubseteq \Sigma_2$.
$\forall a \in \Sigma: \mathcal{L}_{\Sigma_1}$-enforceable$_=$    $ \subseteq \mathcal{L}_{\Sigma_1\xrightarrow{a}\Sigma_2}$-enforceable.
\end{cor}




Corollary \ref{corr:1} indicates that the set of enforceable properties increases monotonically with the capabilities of the monitor. Thus, any effort made to improve the capabilities of the monitor to control its target is rewarded by a augmented set of enforceable security properties. Conversely, if every action from the set $\Sigma$ is in $\mathcal{O}$, only the inviolable property is enforceable.

\begin{cor}
\label{cor:2} Let $\mathcal{L} =\langle\Sigma,\emptyset,\emptyset,\emptyset\rangle, \Property{} \in \mathcal{L}_{\Sigma}$-enforceable $  \Leftrightarrow \Property{} = \Sigma^{\infty}$.
\end{cor}

\section{Alternative Equivalence Relations}\label{sec:equivalence}

The above results apply only to the case of effective$_{=}$ enforcement. This corresponds to the enforcement power of a monitor that  sometimes delays the occurrence of actions in its target, or abort its execution, but does not add additional actions or permanently suppress part of the execution sequence, allowing the remainder of the execution to continue.  Syntactic equality (and subclasses of that relation) is the only equivalence relation that has been extensively studied in the literature. This naturally does not exhaust the enforcement capabilities of monitors.  In this section, we explore alternative equivalence relations (that characterize alternative enforcement mechanisms) and  determine the set of enforceable properties for each.

\subsection{Subword Equivalence and Insertion enforcement}
\label{sect:subword}
The first alternative equivalence relation that we examine is subword equivalence, noted $\cong_{\triangleleft}$. This corresponds to the enforcement power of a monitor which can \textit{add} any action into the execution, but cannot prevent any action from occurring if the target program requests it. For example, the monitor can enforce a property stating that any opened file is eventually closed by adding the missing close file action before the end of the program's execution.  This model is interesting to understand the capabilities of certain types of inline monitors, that are injected into the program in the form of guards or run in parallel with their target, such as those based upon the aspect-oriented programming paradigm \cite{aop,DBLP:journals/toplas/BoddenLH12,DBLP:conf/rv/MeredithR10}.

Let $\sigma, \sigma', \tau\in \Sigma^*$, we write  $\sigma _\tau\cong_{\triangleleft}\sigma'  \Leftrightarrow (\sigma\triangleleft_\Sigma\tau \wedge  \sigma'\triangleleft_\Sigma\tau)$, and designate by enforceable$_{\cong_{\triangleleft}}$ the set of properties that are enforceable when  $_\tau\cong_{\triangleleft}\sigma' $ is the equivalence relation and $\tau$ is the original input. This is a very permissive equivalence relation, which in effect allows the monitor to insert any action or actions into the execution, and consider the transformed execution equivalent to the original  execution  ($\tau$) if all actions performed by the target program are still present.  While this equivalence relation may be too permissive to be realistic, it does allow us to deduce an upper bound to the set of policies enforceable under this paradigm.

We begin by considering the cases that occur when the monitor has only limited control over the action set. When the monitor can only suppress actions or abort the execution, the set of enforceable$_{\cong_{\triangleleft}}$ properties coincides with that of safety properties, since the the monitor is unable to take advantage of the permissiveness of the equivalence relation.
\begin{thm}
Let  $\mathcal{L} =\langle\emptyset,\emptyset,\Sigma,\emptyset\rangle$.   $\mathcal{L}^{\Sigma^\infty}$-enforceable$_{\cong_{\triangleleft}}$ is the set of safety properties.
\end{thm}

Another interesting case occurs when the monitor cannot delay the occurrence of the actions present in the execution sequence, but can only react to them by inserting additional actions afterwards. In other words, any action output by the target program must be immediately accepted, but can be followed by other actions inserted by the monitor. An intuitive lower bound to the set of enforceable properties in this context is the intersection of renewal properties and liveness properties. That is, properties for which any invalid sequence can be corrected into a valid sequence with a finite number of corrective steps.  Additionally, some safety properties and some persistence properties are also enforceable. For example, the property  $\Property_{\neg aa}$ imposing that no two `a' actions occur consecutively is a safety property since an invalid sequence cannot be corrected by the insertion of any subsequent actions. However, the policy is $\mathcal{L}^{\Sigma^\infty}$-enforceable$_{\cong_{\triangleleft}}$ since the monitor can insert any action other than 'a' after the occurrence of each `a'  action to ensure the respect of the security property. What characterizes properties outside the intersection of renewal and liveness that are nonetheless $\mathcal{L}^{\Sigma^\infty}$-enforceable$_{\cong_{\triangleleft}}$  is that there exists a property $\Property'$ in renewal $\cap$ liveness such that the property of interest includes $\Property'$. In the case of the example above, the property  ``any  a action is immediately followed by some action different from a (or the end of sequence token)'' is a renewal property included in $\Property_{\neg aa}$. The monitor can enforce the property $\Property_{\neg aa}$ by enforcing $\Property'$.

\begin{thm}\label{thm:strictInsert}
Let $\Property\subseteq\mathcal{P}(\Sigma^\infty)$ and let $\mathcal{L} =\langle\emptyset,\Sigma,\emptyset,\emptyset\rangle$. If the monitor cannot delay the occurrence of actions performed by the target program, then  $ \Property\in \mathcal{L}^{\Sigma^\infty}$-enforceable$_{\cong_{\triangleleft}}$  $ \Leftrightarrow \exists \Property' : \Property' \subseteq \Property $ and $\Property'$ is infinite renewal $\cap$ liveness.
\end{thm}

Allowing the monitor to insert a finite number of actions  before or after an action taken by the target program increases the set of enforceable  properties further. Consider for example the property \Property$_{os}$  stating that a \textit{write} action only be performed on a previously opened file. The property is a safety property, and falls outside the set of $ \Property\in \mathcal{L}^{\Sigma^\infty}$-enforceable$_{\cong_{\triangleleft}}$ properties defined in theorem \ref{thm:strictInsert} if the number of files is infinite. The property can be  $\mathcal{L}^{\Sigma^\infty}$-enforced$_{\cong_{\triangleleft}}$  by a monitor that inserts the corresponding \textit{open} file action anytime the target program attempts to write to a file that has not yet been opened. More generally, if all actions are in the set $\mathcal{C}$, then a property $\Property$ is enforceable iff there exists a property $\Property'\subseteq\Property$, s.t. $\Property'$ is in renewal and Liveness and for any action $a \in \Sigma$, and for any sequence $\sigma$ in $\Property'$, there is a sequence $\tau$ in the residual of $\Property'$ with regard to $\sigma$ s.t. $a\in acts(\sigma)$. In other words, if for any sequence in the subproperty $\Property'$ and for any possible action $a$, there is a continuation $\tau$ s.t. $\tau$ contains $a$. The monitor can $\mathcal{L}^{\Sigma^\infty}$-enforces$_{\cong_{\triangleleft}}$ such a property by appending a sequence from the residual containing any action requested by the target program.

\begin{thm}
Let  $\mathcal{L} =\langle\emptyset,\Sigma,\emptyset, \emptyset\langle$.   $\mathcal{L}^{\Sigma^\infty}$-enforceable$_{\cong_{\triangleleft}}$ iff there exists a property $\Property' \subseteq   \Property \cap Liveness\cap Renewal:  \forall \sigma\in\Property':  \underset{ \sigma'\in res_{\Property'}(\sigma) }{{\bigcup}}   acts(\sigma')=\Sigma$.
\end{thm}

The upper bound naturally occurs when all actions are controllable. We found it harder to give a specific upper bound to the set of enforceable$_{\cong_{\triangleleft}}$ properties.  Indeed, this set seems to include almost every properties, with the exception of a number of hard to define special cases. Observe that because the subword relation is reflexive, a monitor that enforces$_=$ the property also enforces$_{\cong_{\triangleleft}}$ it. This equivalence relation is so permissive that only a few very particular cases seem unenforceable.

Observe that while this result indicates that the  security properties that are $\mathcal{L}^{\Sigma^\infty}$-enforceable$_{\cong_{\triangleleft}}$ are largely the same as those that are $\mathcal{L}^{\Sigma^\infty}$-enforceable$_{=}$, the manner of enforcement is quite different. $\mathcal{L}^{\Sigma^\infty}$-enforcement$_{=}$ guarantees that  any valid sequence is output \textit{as is}, without any modification by the monitor. For invalid sequences, $\mathcal{L}^{\Sigma^\infty}$-enforceable$_{=}$ ensures that the longest valid prefix is always output \cite{WhatMeanArt}.   $\mathcal{L}^{\Sigma^\infty}$-enforcement$_{\cong_{\triangleleft}}$ does not provide these guarantees. Instead, as seen above, for non-safety properties $\mathcal{L}^{\Sigma^\infty}$-enforcement$_{\cong_{\triangleleft}}$ can ensure that any action present in the original program is eventually output, whether the execution sequence is valid or not. $\mathcal{L}^{\Sigma^\infty}$-enforcement$_{\cong_{\triangleleft}}$ also evidently has a much reduced memory overhead, since this enforcement paradigm does not impose on the monitor that it keep in memory an indefinitely long segment of the execution trace, as $\mathcal{L}^{\Sigma^\infty}$-enforcement$_=$ does. Since  memory constraints were showed in \cite{fong} to significantly affect the set of enforceable properties it is certain that once such constraints are taken into account, some properties will be found to be $\mathcal{L}^{\Sigma^\infty}$-enforceable$_{\cong_{\triangleleft}}$ but not $\mathcal{L}^{\Sigma^\infty}$-enforceable$_=$.

\subsection{Suppression Enforcement}

Another interesting enforcement paradigm is suppression enforcement \cite{MoreEnforce}, which occurs when the monitor can remove potentially malicious actions  performed by the target program but cannot add any action to the execution. In that case, the monitor's enforcement is bounded by the reverse subword equivalence, meaning that the monitor's output is a subword of the original sequence.  This enforcement paradigm is similar to the one described in \cite{runtimeResults}, where the monitor is interposed between the target program and the system. Any action requested by the target program is intercepted by the monitor which must accept or reject it, and allows us to pose an upper bound to this enforcement paradigm.

Let $\sigma, \sigma', \tau\in \Sigma^*$,we write  $\sigma _\tau\cong_{\triangleright}\sigma'  \Leftrightarrow (\tau\triangleleft_\Sigma\sigma \wedge  \tau\triangleleft_\Sigma\sigma')$. As was the case in section \ref{sect:subword}, this is a very permissive equivalence relation, which characterizes the behaviour of a monitor that can potentially \textit{suppress} any action or actions from the execution.  We write $\cong_{\triangleright}$ for this equivalence relation where $\tau$ is the original input.

We begin by determining an upper bound to the set, which occurs when every action is suppressible. In that case, every reasonable property is enforceable, simply by always outputting the empty sequence (or possibly the longest valid prefix). While this may not be a particularly useful enforcement, it does however serve as a useful upper bound to begin reflecting about the capabilities of different types of monitors. It also argues for a stronger notion of transparency, as discussed in\cite{KT2012}, which allows us to reason about the capabilities of monitor in a context that is more similar to that of a real-life monitor. Such monitors would normally be bounded with respect to the alterations that they are allowed to performed on valid and invalid executions alike.

\begin{thm}\label{thm:total}
If $\mathcal{L} =\langle\emptyset,\emptyset,\Sigma,\emptyset\rangle$ then $\mathcal{L}^{\Sigma^\infty}$-enforceable$_{\cong_{\triangleright}}$ = $\mathcal{P}(\Sigma^\infty)$.
\end{thm}

\begin{cor}
Let  $\mathcal{L} =\langle\emptyset,\emptyset,\Sigma,\emptyset\rangle$. $\forall a\in \mathcal{L_{D}}: \mathcal{L}^{\Sigma^\infty}$-enforceable$_{\cong_{\triangleright}}$ = $\mathcal{L}^{\mathcal{S}}_{\mathcal{D}\xrightarrow{a}\mathcal{C}}$-enforceable$_{\cong_{\triangleright}}$.
\end{cor}

\begin{cor}
Let  $\mathcal{L} =\langle\emptyset,\emptyset,\Sigma,\emptyset\rangle$. $\forall \mathcal{S}\subseteq \Sigma^\infty: \mathcal{L}^{\Sigma^\infty}$-enforceable$_{\cong_{\triangleright}}$ =
$\mathcal{L}^{\Sigma^\infty}$-en\-for\-ceable$_{\cong_{\triangleright}}$.
\end{cor}

A more interesting characterization occurs when we consider that some subset of $\Sigma$ is unsuppressible. These may be actions that the monitor lacks the ability to suppress (such as internal system computations) or actions that cannot be deleted without affecting the functionality of the target program. When that is the case, a property is enforceable iff every invalid sequence has a valid prefix ending on an action in $\mathcal{D}$.  In that case, the property of interest includes a safety property that can be enforced by truncation.
\begin{thm}
If $\mathcal{L} =\langle\mathcal{O},\emptyset,\mathcal{D},\emptyset\rangle$ then $\forall\sigma\in \Sigma^\infty: \Property(\sigma)\in \mathcal{L}^{\Sigma^\infty}$-enforceable$_{\cong_{\triangleright}}\Leftrightarrow\exists \Property'\subseteq\Property:\Property' \in$ $\mathcal{L}$-safety.
\end{thm}

\section{Conclusion}\label{sec:conclusion}

In this paper, we reexamined the delimitation of enforceable properties by monitors, and proposed a finer characterization that distinguishes between actions that are only observable, actions which the monitor can delete but not insert into the execution, actions which the monitor can insert in the execution but not suppress if they are already present and completely controllable actions. Our study is a generalization of previous work on the same topic, and provides a finer characterization of the set of security properties that are enforceable by monitors in different contexts.

Additionally, we explored the set of properties that are enforceable by the monitor is given broader latitude to transform valid sequences, rather than be bounded to return a syntactically identical execution sequence  if the original execution is valid. We argue that our results point to the need for an alternative definition of enforcement.

Part of the reason the sets of $\mathcal{L}^{\Sigma^\infty}-$enforceable$_{\cong_{\triangleright}}$ and $\mathcal{L}^{\Sigma^\infty}-$enforceable$_{\cong_{\triangleright}}$ properties are so large is that, in the definition of effective$_{\cong}$ enforcement, the transparency requirement is so weak. This leads to monitors with unusually broad licence to alter invalid sequences in order to correct them.  For monitors with broad capabilities to add and remove actions from the execution sequence, the desired behaviour of real-life security policy enforcement mechanism would be more accurately characterized by a more constraining definition of \textit{enforcement}.  For example, a practical suppression monitor should be bounded to remove from an invalid execution sequence only those actions that violate the security policy. Any valid behaviour present in an otherwise invalid sequence should be preserved.

In the future, we would like to consider that the cost of inserting or deleting an action may not be the same for all actions.  This would allow us to contrast different enforcement strategies for the same security property. A lattice-based framework is well-suited to model such a restriction.

\bibliographystyle{abbrv}
\bibliography{report}
\newpage
\appendix
\section{Proofs} \label{App:AppendixA}
\setcounter{thm}{0}
\setcounter{cor}{0}
\begin{thm}
\begin{multline*}
\Property{} \in \mathcal{L}^{\Sigma^\infty}-\mbox{enforceable}_{=}  \Leftrightarrow \forall \sigma\in \Sigma^\infty:\\
(\Property(\sigma)\Leftrightarrow \forall\sigma'\preceq\sigma:\exists \tau\preceq\sigma:\sigma'\preceq\tau\wedge\Property(\tau)\wedge (\forall\tau'\preceq\tau:\neg\Property(\tau')\Rightarrow \tau'.last\in \mathcal{C})\vee \\
(\exists \tau;a\preceq\sigma: a\in \langle \mathcal{D} \cup \mathcal{C} \rangle \wedge \forall\tau'\preceq\tau;a:\neg\Property(\tau')\Rightarrow \tau'.last\in \mathcal{C} \wedge
(\forall \tau'\succeq\tau;a:\Property(\tau')\Rightarrow \\
\tau'=\sigma  \vee  (\Property(\tau)\wedge \neg\exists \tau'\succeq \tau:\Property(\tau') \wedge  acts(\sigma \backslash \tau' ;a) \in \langle C \cup I \rangle ) )) )
\end{multline*}
\end{thm}
\begin{proof}
(if direction) We construct an edit-automaton $\mathcal{A}= \langle Q, q_0, \delta\rangle$ that effectively$_{=}$ enforces the property \Property. This automata is defined as follows:
\begin{itemize}
  \item $Q= \Sigma^* \times \Sigma^*$, i.e., each state consists in a pair of finite sequences, the sequence output so far and the suffix of the input currently suppressed.
  \item $q_0=(\epsilon,\epsilon)$,
  \item the transition function $\delta:Q\times \Sigma\times Q$ is defined as \\
  $\delta((\sigma_o,\sigma_s),a)=\left\{
                                    \begin{array}{ll}
                                      (\sigma_o,\sigma_s;a), & \hbox{if } \neg\Property(\sigma_o;\sigma_s;a) \wedge \gamma(\sigma_o;\sigma_s;a)=\bot \\
                                      (\sigma_o;\sigma_s;a,\epsilon), & \hbox{}\Property(\sigma_o;\sigma_s;a)\\
                                       (\sigma_o;a',\epsilon), & \hbox{} \neg\Property(\sigma_o;\sigma_s;a) \wedge \gamma(\sigma_o;\sigma_s;a)=a'.
                                    \end{array}
                                  \right.$
    \\where $\gamma:\Sigma^*\times \Sigma\cup \bot$ is a function defined as follows \\
    $\gamma(\sigma)=\left\{
                    \begin{array}{ll}
                      a, & \hbox{}\exists a;\tau\in\Sigma^\infty: \forall \sigma'\succeq\sigma: \Property\sigma' \Rightarrow \sigma'=a;\tau\wedge acts(\tau)\in \{\mathcal{I,C}\} ; \\
                      \bot, & \hbox{otherwise.}
                    \end{array}
                  \right.
    $
    Informally, this transition function states that the monitor suppresses the execution until a valid prefix is encountered, at which point it outputs the suffix of the execution it has suppressed so far. Function $\gamma$ detects the occurrence of the corner case described above.

\end{itemize}
Let $\sigma \in \Sigma^\infty$ be the input sequence and let $\sigma[..i]$ be the prefix of the input sequence processed at step $i$, and let $q=\langle\sigma^i_o, \sigma^i_s\rangle$ be the state of $\mathcal{A}$  reached after processing $\sigma[..i]$.  The automaton maintains the following invariants.
At step $i$, (1) $\sigma_o$ has been output and this output is valid, (2) the output is transparent, i.e. $\sigma[..i]\in \Property \Rightarrow  \sigma[..i] \in \sigma \vee  \exists \sigma'\preceq\sigma :\forall \sigma';\tau\succeq\sigma':\Property(\sigma';\tau)\Rightarrow \sigma';\tau=\sigma $ (3)$\gamma(\sigma[..i])=\bot \Rightarrow \sigma[..i]=\sigma^i_o;\sigma^i_s$ and (4) the automaton $\mathcal{A}$ never manipulates and action in a disallowed manner.
The third invariant ensures that the automata does not, for instance, suppress an unsuppressible action, or abort the execution on an observable action.

The invariants hold initially since $\sigma[..0]=\epsilon$ and the automaton is in state $\langle\epsilon,\epsilon\rangle$.  Let us assume that INV(i) holds and let $a$ be the next input action. We show that INV(i+1) holds where $\sigma_{i+1} = \sigma[..1];a$.

There are three cases to consider.\\
\begin{enumerate}
  \item $\sigma[..i];a$ is invalid, and there are multiple valid extensions.   In this case the automaton suppresses the input sequence. The automaton enters state $\langle\sigma^{i+1}_o, \sigma^{i+1}_s\rangle$ where $\sigma^{i+1}_o = \sigma^{i}_o$ and $\sigma^{i+1}_s= \sigma^{i}_s;a$. That $\sigma^{i+1}_o \in \Property$  holds from the induction hypothesis. Likewise, the transparency requirement holds trivially since $\neg\Property(\sigma[..i];a)$.  The first conjunct of theorem \ref{thm:Master} ensures that $a$ is controllable.

  \item $\Property(\sigma[..i];a)$. In this case the automaton outputs $\sigma^i_o;a$. By induction we have that $\sigma^i_o;\sigma^i_s;a = \sigma[..i];a$ and thus $\Property(\sigma^{i+1}_o;)$. $\sigma^i_s= \epsilon;$, which means the third part of the invariant holds trivially. Finally, since for in previous states, an action could only have been suppressed if it was controllable, we have that every action in $\sigma^i_o$ is controllable.

  \item $\sigma[..i];a$ has a single valid extension. This includes the case where the only valid extension is $\epsilon$ and the execution is aborted. In this case the automaton enters a loop in which it outputs the actions of the only valid extension one  by one. By the theorem \ref{thm:Master} we have that the output of the monitor is valid and insertable or controlable. Such an action in known to exist by the final conjunct of the theorem \ref{thm:Master}. INV(i+1)(3) holds trivially since $\gamma(\sigma[..i])\neq\bot$.
\end{enumerate}

Since the invariant INV holds at each step, the output is valid, transparent and the monitor does not behave in a manner inconsistent with the limitations on its capabilities. \\
(else-if direction)\\
By negation of theorem \ref{thm:Master}, we have that a property is unenforceable iff it contains at least one valid sequence $\sigma$ meeting the following two properties \\

\begin{enumerate}
 \item  $ \exists \sigma'\preceq \sigma : \neg \Property(\sigma') \wedge \sigma'.last\notin \mathcal{C}$

 \item $\exists \sigma' \preceq \sigma  : \exists \tau\in Property: \tau\succeq\sigma' \wedge \tau \neq\sigma \vee \exists a \in acts(\tau) : a \notin \langle\mathcal{I} \cup \mathcal{C}\rangle$

\end{enumerate}
Informally, the property is unenforceable iff there exists a valid sequence with an invalid prefix that is not in the controllable set, that that valid sequence either is not in the corner case described in section 3 or is not comprised of insertable of controllable actions. When the monitor encounters such a prefix, it cannot accept it since it may be the end of an invalid sequence, but it cannot suppress it since,  the input may have a valid continuation, which the monitor would be unable to re-insert.
\end{proof}

\setcounter{thm}{3}

\begin{thm}
(from \cite{enforceable}) If $\mathcal{L} =\langle\emptyset,\emptyset,\Sigma,\emptyset\rangle$ then $\mathcal{L}^{\Sigma^\infty}$enforceable$_{=}$ is  Safety.
\end{thm}
\begin{proof}
$\Property \in \mathcal{L}^{\Sigma^\infty}$-enforceable$_{=} \Leftrightarrow \forall \sigma\in \Sigma^\infty:\\
(\Property(\sigma)\Leftrightarrow \forall\sigma'\preceq\sigma:\exists \tau\preceq\sigma:\sigma'\preceq\tau\wedge\Property(\tau)\wedge (\forall\tau'\preceq\tau:\neg\Property(\tau')\Rightarrow \tau'.last\in \mathcal{C})\vee \\ (\exists \tau;a\preceq\sigma: a\in \langle D \cup C \rangle \wedge \forall\tau'\preceq\tau;a:\neg\Property(\tau')\Rightarrow \tau'.last\in \mathcal{C} \wedge \\ (\forall \tau'\succeq\tau;a:\Property(\tau')\Rightarrow \tau'=\sigma  \vee  (\Property(\tau)\wedge \neg\exists \tau'\succeq \tau:\Property(\tau') )) ) )$ \\

\hspace*{\fill}$\langle \forall a \in \Sigma : a \in \mathcal{D}\rangle $\\

$\Property \in \mathcal{L}^{\Sigma^\infty}$-enforceable$_{=} \Leftrightarrow \forall \sigma\in \Sigma^\infty:\\
(\Property(\sigma)\Leftrightarrow \forall\sigma'\preceq\sigma:\exists \tau\preceq\sigma:\sigma'\preceq\tau\wedge\Property(\tau)\wedge (\forall\tau'\preceq\tau:\Property(\tau)\vee \\ (\exists \tau;a\preceq\sigma: \forall\tau'\preceq\tau;a:\neg\Property(\tau')\Rightarrow \bot \wedge \\ (\forall \tau'\succeq\tau;a:\Property(\tau')\Rightarrow \tau'=\sigma  \vee  (\Property(\tau)\wedge \neg\exists \tau'\succeq \tau:\Property(\tau') )) ) )$\\

\hspace*{\fill}$\langle \bot \wedge A  = \bot \rangle $\\

$\Property \in \mathcal{L}^{\Sigma^\infty}$-enforceable$_{=} \Leftrightarrow \forall \sigma\in \Sigma^\infty:\\
(\Property(\sigma)\Leftrightarrow \forall\sigma'\preceq\sigma:\exists \tau\preceq\sigma:\sigma'\preceq\tau\wedge (\forall\tau'\preceq\tau:\Property(\tau)\vee \\ (\exists \tau;a\preceq\sigma: \forall\tau'\preceq\tau;a:\Property(\tau') \\   \vee  (\Property(\tau)\wedge \neg\exists \tau'\succeq \tau:\Property(\tau') )) ) )$\\

$\Property \in \mathcal{L}^{\Sigma^\infty}$-enforceable$_{=} \Leftrightarrow \forall \sigma\in \Sigma^\infty:\\
(\Property(\sigma)\Leftrightarrow \forall\sigma'\preceq\sigma: \Property(\sigma) ) )$\\

\hspace*{\fill}$\langle$ definition of Safety $\rangle $\\

$\Property \in \mathcal{L}^{\Sigma^\infty}$-enforceable$_{=} \Leftrightarrow  \Property \in$ safety \\
\end{proof}

\begin{thm}
(from \cite{nonsafetyJournal}) If $\mathcal{L} =\langle\emptyset,\emptyset,\emptyset,\Sigma\rangle$ then $\mathcal{L}^{\Sigma^\infty}$enforceable$_{=}$ is Infinite Renewal or the corner case.
\end{thm}
\begin{proof}
$\Property \in \mathcal{L}^{\Sigma^\infty}$-enforceable$_{=}$ $ \Leftrightarrow \forall \sigma\in \Sigma^\infty:\\
(\Property(\sigma)\Leftrightarrow \forall\sigma'\preceq\sigma:\exists \tau\preceq\sigma:\sigma'\preceq\tau\wedge\Property(\tau)\wedge (\forall\tau'\preceq\tau:\neg\Property(\tau')\Rightarrow \tau'.last\in \mathcal{C})\vee \\ (\exists \tau;a\preceq\sigma: a\in \langle D \cup C \rangle \wedge \forall\tau'\preceq\tau;a:\neg\Property(\tau')\Rightarrow \tau'.last\in \mathcal{C} \wedge \\ (\forall \tau'\succeq\tau;a:\Property(\tau')\Rightarrow \tau'=\sigma  \vee  (\Property(\tau)\wedge \neg\exists \tau'\succeq \tau:\Property(\tau') )) ) )$ \\

\hspace*{\fill}$\langle \forall a \in \Sigma : a \in \mathcal{C}\rangle $\\

$\Property \in \mathcal{L}^{\Sigma^\infty}$-enforceable$_{=}$ $ \Leftrightarrow \forall \sigma\in \Sigma^\infty:\\
(\Property(\sigma)\Leftrightarrow \forall\sigma'\preceq\sigma:\exists \tau\preceq\sigma:\sigma'\preceq\tau\wedge\Property(\tau)\wedge \top )\vee \\
(\exists \tau;a\preceq\sigma:   \forall\tau'\preceq\tau;a:\neg\Property(\tau')\Rightarrow \top \wedge \\ (\forall \tau'\succeq\tau;a:\Property(\tau')\Rightarrow \tau'=\sigma  \vee  (\Property(\tau)\wedge \neg\exists \tau'\succeq \tau:\Property(\tau') )) ) )$ \\

\hspace*{\fill}$\langle \top \wedge A = A\rangle $\\

$\Property \in \mathcal{L}^{\Sigma^\infty}$-enforceable$_{=}$ $ \Leftrightarrow \forall \sigma\in \Sigma^\infty:\\
(\Property(\sigma)\Leftrightarrow \forall\sigma'\preceq\sigma:\exists \tau\preceq\sigma:\sigma'\preceq\tau\wedge\Property(\tau))\vee \\
(\exists \tau;a\preceq\sigma:   \forall\tau'\preceq\tau;a: \\ (\forall \tau'\succeq\tau;a:\Property(\tau')\Rightarrow \tau'=\sigma  \vee  (\Property(\tau)\wedge \neg\exists \tau'\succeq \tau:\Property(\tau') ) )) )$ \\

\hspace*{\fill}$\langle \forall \sigma\in \Sigma^\infty: \Property(\sigma)\Leftrightarrow (\forall\sigma'\preceq\sigma:\exists \tau\preceq\sigma:\sigma'\preceq\tau\wedge\Property(\tau)) \Rightarrow
\exists \tau;a\preceq\sigma:  \forall \tau'\succeq\tau;a:  \neg\Property(\tau')\Rightarrow  (\Property(\tau)\wedge \neg\exists \tau'\succeq \tau:\Property(\tau')) \rangle $\\

$\Property \in \mathcal{L}^{\Sigma^\infty}$-enforceable$_{=}$ $ \Leftrightarrow \forall \sigma\in \Sigma^\infty:\\
(\Property(\sigma)\Leftrightarrow \forall\sigma'\preceq\sigma:\exists \tau\preceq\sigma:\sigma'\preceq\tau\wedge\Property(\tau))\vee \\
(\exists \tau;a\preceq\sigma:   \forall\tau'\preceq\tau;a: \\ (\forall \tau'\succeq\tau;a:\Property(\tau')\Rightarrow \tau'=\sigma  )) )$ \\

\hspace*{\fill}$\langle$ definition of renewal and of the corner case $\rangle $\\

$\Property \in \mathcal{L}^{\Sigma^\infty}$-enforceable$_{=} \Leftrightarrow  \Property \in$ renewal or the corner case  \\
\end{proof}


\begin{cor}\label{corr:1}
Let $\mathcal{L}$ be a lattice over a set of actions $\Sigma$ as described above and let $\Sigma_1,\Sigma_2\in\{\mathcal{O},\mathcal{D},\mathcal{I},\mathcal{C}\}$ and $\Sigma_1 \sqsubseteq \Sigma_2$.
$\forall a \in \Sigma: \mathcal{L}_{\Sigma_1}$-enforceable$_=$    $ \subseteq \mathcal{L}_{\Sigma_1\xrightarrow{a}\Sigma_2}$-enforceable.
\end{cor}
\begin{proof}
Any property that is $\mathcal{L}_{\Sigma_1}$-enforceable is trivially $\mathcal{L}_{\Sigma_1\xrightarrow{a}\Sigma_2}$-enforceable.\\

Let $\Sigma_1=\mathcal{O}$ and  $\Sigma_2\in\{\mathcal{I},\mathcal{C} \}$.
The property $\Property(\sigma)\Leftrightarrow \exists \tau\in \Sigma^\infty : \sigma = a;tau$, which states that any valid execution must begin with a distinguished action $a$ is not $\mathcal{L}_{\Sigma_1}$-enforceable since the monitor cannot correct an invalid sequence by adding the missing initial action. It is $\mathcal{L}_{\Sigma_1\xrightarrow{a}\Sigma_2}$-enforceable.

Let  $\Sigma_1=\mathcal{D}$ and  $\Sigma_2 = \mathcal{C}$.
The property $\Property(\sigma)\Leftrightarrow \sigma = aa$ is not $\mathcal{L}_{\Sigma_1}$-enforceable since when faced with a single $a$, the monitor can neither suppress it (which would make it impossible to return a valid syntactically equal sequence if the next action is also $a$), nor output it, since the execution would be then be irremediably invalid in all other cases. The property is  $\mathcal{L}_{\Sigma_1\xrightarrow{a}\Sigma_2}$-enforceable.

Let  $\Sigma_1=\mathcal{O}$ and  $\Sigma_2 = \mathcal{D}$  or $\Sigma_1=\mathcal{I}$ and  $\Sigma_2 = \mathcal{C}$.
The property $\Property(\sigma)\Leftrightarrow a\notin acts(\sigma)$ is not $\mathcal{L}_{\Sigma_1}$-enforceable since the monitor lacks the ability to suppress invalid $a$ actions.  It is $\mathcal{L}_{\Sigma_1\xrightarrow{a}\Sigma_2}$-enforceable by simple suppression of these actions.
\end{proof}

\begin{cor}
\label{cor:2} Let $\mathcal{L} =\langle\Sigma,\emptyset,\emptyset,\emptyset\rangle, \Property{} \in \mathcal{L}_{\Sigma}$-enforceable $  \Leftrightarrow \Property{} = \Sigma^{\infty}$.
\end{cor}
\begin{proof}
Corollary \ref{cor:2} follows immediately from Theorem \ref{thm:Master}.
\end{proof}


\begin{thm}
Let  $\mathcal{L} =\langle\emptyset,\emptyset,\Sigma,\emptyset\rangle$.   $\mathcal{L}^{\Sigma^\infty}$-enforceable$_{\cong_{\triangleleft}}$ is the set of safety properties.
\end{thm}
\begin{proof}
Trivial. Since the monitor lacks the ability to insert any action into the execution, it behaves like a suppression automaton \cite{MoreEnforce}. The equivalence relation imposes that any action present in the input must also be present in the output. Since the monitor can only allow an action or terminate the execution, any enforceable property is necessarily prefix closed.
\end{proof}


\begin{thm}
Let $\Property\subseteq\mathcal{P}(\Sigma^\infty)$ and let $\mathcal{L} =\langle\emptyset,\Sigma,\emptyset,\emptyset\rangle$. If the monitor cannot delay the occurrence of actions performed by the target program, then  $ \Property\in \mathcal{L}^{\Sigma^\infty}$-enforceable$_{\cong_{\triangleleft}}$  $ \Leftrightarrow \exists \Property' : \Property' \subseteq \Property{} $ and $\Property'$ is infinite renewal $\cap$ liveness.
\end{thm}
\begin{proof}
(if direction)
Let \Property' be property in Renewal $\cap$ Liveness s.t. \Property' $\subseteq$ \Property\ and let $\mathcal{L} =\langle\emptyset,\Sigma,\emptyset,\emptyset\rangle$. A monitor can $\mathcal{L}^{\Sigma^\infty}_{\cong_{\triangleleft}}$-enforce \Property' as follows: let $\tau\in\Sigma^*$ be the current output so far and let $a\in\Sigma$  be the next program action in the execution. Since \Property' is in Renewal $\cap$ Liveness there necessarily exists a finite sequence $\tau'\in\sigma^*$ s.t. $\tau;a;\tau'\in \Property'$, and the monitor can $\mathcal{L}^{\Sigma^\infty}_{\cong_{\triangleleft}}$-enforce \Property' by outputting this sequence. Since \Property' $\subseteq$ \Property, the  monitor's output is correct w.r.t. \Property\ and since $\tau;a;\tau'\triangleleft_\Sigma a $ the output is transparent. \\
(else-if direction)
A \Property\  not in the intersection of Renewal $\cap$ Liveness can be safety properties or persistence properties.  Safety properties are properties for which a violation of the security policy is irremediable. As such, no suffix can be added by the monitor to the invalid sequence to correct it, thus $\mathcal{L}^{\Sigma^\infty}_{\cong_{\triangleleft}}$-enforcing the property. Likewise, persistence properties include infinite invalid sequence with infinitely many valid prefixes, and infinite valid prefixes with only finitely many valid prefixes. In the former case, the monitor cannot enforce the property because even though its output will continuously be valid, the overall execution will violate the property. In the latter case, the monitor will eventually reach a finite execution $\sigma$ which has no finite valid extension.  If the monitor outputs a valid infinite extension, subsequent action output by the program will not be present in the monitor's output, violating the transparency requirement.  Since there does not exists a property \Property' $\subseteq$ \Property\ s.t. \Property' is in Renewal $\cap$ Liveness, the monitor cannot avoid reaching one of the three cases described above.
\end{proof}


\begin{thm}
Let  $\mathcal{L} =\langle\emptyset,\Sigma,\emptyset, \emptyset\rangle$.   $\mathcal{L}^{\Sigma^\infty}$-enforceable$_{\cong_{\triangleleft}}$ iff there exists a property $\Property' \subseteq   \Property \cap Liveness\cap Renewal:  \forall \sigma\in\Property':  \underset{ \sigma'\in res_{\Property'}(\sigma) }{{\bigcup}}   acts(\sigma')=\Sigma$.
\end{thm}
\begin{proof}
(if direction) Let \Property\ and \Property'\ be a properties as described above. Let $\sigma\in \Sigma^*$ be the input sequence so far and let $a\in\Sigma$ be the next program action in the execution.  By definition, there exists a sequence $\tau$ s.t. $\sigma;\tau\in \Property'$  (satisfying correctness since \Property' $\subseteq$ \Property) and $a \in acts(\tau)$ (satisfying transparency).\\
(else-if direction) If the condition above does no hold, then by definition there exists a sequence $\tau\in\Sigma^*$ such that $\tau$ and an action $a\in\Sigma$ such that $\tau$ has no valid finite extension containing $a$.  This makes it impossible to enforce for a monitor to enforce the property in a correct and transparent manner if $\tau$ is the input sequence.
\end{proof}

\begin{thm}
If $\mathcal{L} =\langle\emptyset,\emptyset,\Sigma,\emptyset\rangle$ then $\mathcal{L}^{\Sigma^\infty}$-enforceable$_{\cong_{\triangleright}}$ = $\mathcal{P}(\Sigma^\infty)$.
\end{thm}
\begin{proof}
Since the empty sequence $\epsilon$ is always valid, and equivalent to every possible sequence, the monitor can $\mathcal{L}^{\Sigma^\infty}$-enforce$_{\cong_{\triangleright}}$ any property by always outputting that sequence.
\end{proof}

\begin{cor}
Let  $\mathcal{L} =\langle\emptyset,\emptyset,\Sigma,\emptyset\rangle$. $\forall a\in \mathcal{L_{D}}: \mathcal{L}^{\Sigma^\infty}$-enforceable$_{\cong_{\triangleright}}$ = $\mathcal{L}^{\mathcal{S}}_{\mathcal{D}\xrightarrow{a}\mathcal{C}}$-enforceable$_{\cong_{\triangleright}}$.
\end{cor}
\begin{proof}
Immediate from theorem \ref{thm:total}.
\end{proof}

\begin{cor}
Let  $\mathcal{L} =\langle\emptyset,\emptyset,\Sigma,\emptyset\rangle$. $\forall \mathcal{S}\subseteq \Sigma^\infty: \mathcal{L}^{\Sigma^\infty}$-enforceable$_{\cong_{\triangleright}}$ =
$\mathcal{L}^{\Sigma^\infty}$-en\-for\-ceable$_{\cong_{\triangleright}}$.
\end{cor}
\begin{proof}
Immediate from theorem \ref{thm:total}.
\end{proof}

\setcounter{thm}{9}
\begin{thm}
If $\mathcal{L} =\langle\mathcal{O},\emptyset,\mathcal{D},\emptyset\rangle$ then $\forall\sigma\in \Sigma^\infty: \Property(\sigma)\in \mathcal{L}^{\Sigma^\infty}$-enforceable$_{\cong_{\triangleright}} \Leftrightarrow  \exists \Property'\subseteq\Property:\Property' \in$ $\mathcal{L}$-safety.
\end{thm}
\begin{proof}
(if direction)The property \Property\ can be enforced by enforcing the property $\Property' \subseteq\Property$. Since \Property' is a safety property, it can be enforced by truncation. Since \Property' is a subset of \Property, enforcing \Property' satisfy the correctness requirement for \Property. Enforcing the property by truncation guarantees that the transparency requirement is respected.

(else-if direction)Let \Property\ be a property for which there does not exist a property \Property' such that  $\Property'\subseteq\Property$  and \Property' $\in$ $\mathcal{L}$-safety. There must exist a sequence $\sigma\notin\Property$ such that $\neq\exists \tau\preceq\sigma:\tau\in\Property\wedge \tau.last\in\{\mathcal{D}\cup\mathcal{C}\}\wedge \neq\exists\tau'\succeq\tau:\Property(\tau')$ . If $\tau$ this is the input sequence, the monitor cannot enforce the property since there is no valid prefix upon which it can abort the execution.
\end{proof}

\end{document}